\documentclass[11pt,english]{article}
\usepackage[table]{xcolor}
\usepackage{lineno}
\usepackage{amsmath,amssymb,amsthm}
\usepackage{dsfont}
\usepackage{multicol}
\usepackage[margin=1in]{geometry}
\usepackage{graphicx,color}
\usepackage{enumitem}
\usepackage{fullpage}
\usepackage[noblocks]{authblk}
\usepackage{tcolorbox}
\usepackage{babel}
\usepackage{wrapfig}
\usepackage{mdframed} 
\usepackage{mathtools}
\usepackage{floatrow}
\usepackage{multirow}
\usepackage{tablefootnote}
\usepackage[leftcaption]{sidecap}
\usepackage{hhline}
\usepackage{pifont}

\usepackage[ruled,linesnumbered,vlined]{algorithm2e}
\usepackage{tablefootnote}
\usepackage{thm-restate}
\usepackage{bbding}
\usepackage{pifont}
\usepackage{nicefrac}

\renewcommand{\arraystretch}{1.2}

\definecolor{Darkblue}{rgb}{0,0,0.4}
\definecolor{Brown}{cmyk}{0,0.61,1.,0.60}
\definecolor{Purple}{cmyk}{0.45,0.86,0,0}
\definecolor{Darkgreen}{rgb}{0.133,0.543,0.133}

\usepackage[colorlinks,linkcolor=Darkblue,filecolor=blue,citecolor=blue,urlcolor=Darkblue,pagebackref]{hyperref}
\usepackage[nameinlink]{cleveref}

\usepackage[colorinlistoftodos,prependcaption,textsize=tiny]{todonotes}

\newcommand{\atodoAfter}[1]{}
\newcommand{\aidea}[1]{}

\newcommand{\submit}[1]{}

\definecolor{BrickRed}{rgb}{.72,0,0}
\def\EMPH#1{\emph{\textcolor{BrickRed} {#1}}}
\newcommand{\brick}[1]{{\color{BrickRed}#1}}

\newif\ifdraft 
\draftfalse

\newcommand{\namedref}[2]{\hyperref[#2]{#1~\ref*{#2}}}
\newcommand{\propref}[1]{\hyperref[#1]{property~(\ref*{#1})}}

\newcommand{\tw}{{\rm tw}}

\newlength{\Oldarrayrulewidth}

\crefname{resultrow}{Result}{Results}
\Crefname{resultrow}{Result}{Results}

\newtheorem{theorem}{Theorem}
\newtheorem{lemma}{Lemma}
\newtheorem{definition}{Definition}

\newtheorem{remark}{Remark}

\newcommand{\polylog}{\mathrm{polylog}}

\newcommand{\R}{\mathbb{R}}

\def\cP{\ensuremath{\mathcal{P}}}
\def\cQ{\ensuremath{\mathcal{Q}}}

\def\cS{\ensuremath{\mathcal{S}}}

\definecolor{forestgreen}{rgb}{0.13, 0.55, 0.13}

\SetCommentSty{mycommfont}

\def\eps{\varepsilon}

\DeclareMathAlphabet{\mathpzc}{OT1}{pzc}{m}{it}
\newcommand{\etal}{{\em et al. \xspace}}

\newcommand{\rt}{\mbox{\rm rt}}

\newlength{\dhatheight}

\newcommand {\ignore} [1] {}

\newcommand{\initOneLiners}{%
	\setlength{\itemsep}{0.2pt}
	\setlength{\parsep }{0.2pt}
	\setlength{\topsep }{0.2pt}
}

\usepackage{sansmath}

\title{DAG Covers for Structured Graphs: The Steiner
Point Effect}

\author{Sujoy Bhore\thanks{Indian Institute of Technology Bombay. Email: \texttt{sujoy@cse.iitb.ac.in}. Work supported in part by ANRF ARG-MATRICS, Grant 002465.} \quad
Hsien-Chih Chang\thanks{Dartmouth College. Email: \texttt{hsien-chih.chang@dartmouth.edu}. Supported by the U.S. National Science Foundation Grant No. CCF-2443017.} \quad 
Jonathan Conroy\thanks{Dartmouth College. Email: \texttt{jonathan.conroy.gr@dartmouth.edu}. Supported by the U.S. National Science Foundation Grant No. CCF-2443017.}
\quad 
Arnold Filtser\thanks{Bar-Ilan University. Email: \texttt{arnold.filtser@biu.ac.il}. This research was supported by the ISRAEL SCIENCE FOUNDATION (grant No. 1042/22).}
\quad
Eunjin Oh\thanks{POSTECH. Email: \texttt{eunjin.oh@postech.ac.kr}.}\quad 
Nicole Wein\thanks{University of Michigan, Ann Arbor. Email: \texttt{nswein@umich.edu}.} \quad
Da Wei Zheng\thanks{Institute of Science and Technology Austria. Email: \texttt{dzheng@ista.ac.at}.}}

\date{}

\begin{document}
\maketitle
\begin{abstract}
Given a weighted digraph $G$, a $(t,g,\mu)$-DAG cover is a collection of $g$ dominating DAGs $D_1,\dots,D_g$ such that all distances are approximately preserved: for every pair $(u,v)$ of vertices, $\min_id_{D_i}(u,v)\le t\cdot d_{G}(u,v)$, and the total number of non-$G$ edges is bounded by $|(\cup_i D_i)\setminus G|\le \mu$.
Assadi, Hoppenworth, and Wein [STOC 25] and Filtser [SODA 26] studied DAG covers for general digraphs.
This paper initiates the study of \emph{Steiner} DAG cover, where the DAGs are allowed to contain Steiner points. 

We obtain Steiner DAG covers on the important classes of planar digraphs and low-treewidth digraphs. Specifically, we show that any digraph with treewidth $\tw$ admits a $(1,2,\tilde{O}(n\cdot\tw))$-Steiner DAG cover. For planar digraphs we provide a $(1+\varepsilon,2,\tilde{O}_\varepsilon(n))$-Steiner DAG cover.

We also demonstrate a stark difference between Steiner and non-Steiner DAG covers. As a lower bound, we show that any non-Steiner DAG cover for graphs with treewidth $1$ with stretch $t<2$ and sub-quadratic number of extra edges requires $\Omega(\log n)$ DAGs.
\end{abstract}

\pagenumbering{gobble}
\newpage

\setcounter{secnumdepth}{5}
\setcounter{tocdepth}{3} 
\tableofcontents
\thispagestyle{empty}
    \newpage
    \pagenumbering{arabic}

\section{Introduction}

Tree covers are fundamental graph primitives with wide-ranging applications in network routing, approximate distance oracles, spanners, and related problems. They offer sparse and structurally simple representations of large graphs while retaining key distance relationships. 
Formally, a\emph{\ stretch-$t$ tree cover} of an edge-weighted undirected graph $G$ is a collection of weighted trees $\mathcal{T}$ on vertex set $V$ such that for every pair $u,v\in V$, (1) for every tree $T\in \mathcal{T}$, we have $d_G(u,v)~\le~d_T(u,v)$, and (2) there exists some $T \in \mathcal{T}$ such that
$d_T(u,v) \;\le\; t \cdot d_G(u,v)$. 

Tree covers were introduced 30 years ago by Arya, Das, Mount, Salowe, and Smid~\cite{ADMSS95} in the context of Euclidean spaces and have been extensively studied ever since. In $\mathbb{R}^d$, they established stretch-$(1+\eps)$ tree cover with 
$O(\varepsilon^{-d} \log(1/\varepsilon))$ trees, which was later improved to 
$O(\varepsilon^{-d+1} \log(1/\varepsilon))$  \cite{ES15}. For general graphs, tree covers were studied implicitly in~\cite{awerbuch1994buffer, MR1157580}, and then in their seminal work Thorup and Zwick~\cite{TZ05} showed a tree cover with stretch $2k-1$ and 
$O\!\left(n^{1/k} \log n\right)$ trees, for any integer $k\geq 1$. 
Tree covers have found many of their applications in graphs with geometric or topological structure.  In particular, tree covers have been extensively studied and applied in the context of Euclidean metrics~\cite{ADMSS95,CCLMST24socg, ES15}, doubling 
metrics~\cite{BFN22,CGMZ05,FGN24,FL22}, 
planar, bounded treewidth, and minor-free 
graphs~\cite{bhore2025spanners,CCLMST23Planar,chang2024shortcut, elkin2025spanning, GKR04}.

In sharp contrast, directed graphs admit a far weaker repertoire of usable distance-preserving structures. Many sparsification primitives that play a central role in undirected graphs are either provably impossible to obtain in directed settings (e.g., sparse spanners~\cite{ABSHJKS20} or efficient distance oracles~\cite{TZ05}) or remain substantially less understood (e.g., distance preservers and hopsets). 
For directed hopsets, recent work has made substantial progress, yet polynomial gaps between upper and lower bounds persist, in contrast to the near-tight landscape for undirected counterparts 
(see, e.g.,~\cite{bernstein2023closing, bodwin2023folklore, elkin2019linear,  hoppenworth2025new, huang2019thorup, kogan2022new, williams2024simpler}).

Directed acyclic graphs (DAGs) are natural analogs of trees in directed graphs and often distance problems become more tractable on DAGs than general directed graphs. For instance, computing single-source shortest paths with negative weights is somewhat trivial in DAGs and achievable in linear time, yet for general directed graphs the problem remains extremely difficult. Distance preservers also exhibit a gap: the best bounds known for DAGs are strictly stronger than those for general directed graphs~\cite{bodwin2017linear}. Even for hopsets, where upper bounds for the simpler shortcut-set problem extend readily in DAGs, analogous results are more elusive for general directed 
graphs~\cite{bernstein2023closing,kogan2022new}. 

Motivated by this phenomenon, in a recent work~\cite{AHW25}, Assadi, Hoppenworth, and Wein introduced the notion of \emph{DAG cover} analogous to a tree cover: informally, given a weighted directed graph $G$, a \emph{DAG cover} is a collection of DAGs such that no DAG underestimates any distance, and every distance $d_G(u,v)$ is approximated by at least one of the DAGs (see \Cref{def:DagCover}). 
They observe that a directed cycle constitutes a trivial lower bound if each DAG is required to be a subgraph of $G$. In this case $n$ DAGs are needed even to preserve reachability. (From the upper bounds side, $n$ DAGs can always trivially preserve distances exactly by taking a shortest paths tree from each vertex.) 
Thus, it is necessary to allow the addition of edges not from $G$. However, allowing quadratically many additional edges trivializes the problem: only two dense DAGs with opposite topological orders yield exact distances by adding all $\Theta(n^2)$ edges consistent with the ordering (with appropriate edge weights). Thus, the goal becomes to trade-off between three parameters: the number of DAGs, the stretch, and the number of additional edges. Building upon~\cite{AHW25}, Filtser~\cite{Fil25dags} proved the best-known upper bound for general directed graphs: with $O(\log n)$ DAGs, it is possible to achieve stretch $\tilde{O}(\log n)$, with the addition of $\tilde{O}(m)$ edges. Building off this line of work, very recently Haeupler, Jiang, and Saranurak~\cite{proj} defined the related notion of a \emph{DAG projection}, which they used to reduce a number of problems from general graphs to DAGs. See Section~\ref{sec:related} for additional related work.

\subsection{Our Setting: Steiner DAG Covers for Structured Graphs}
Given that DAG covers necessarily require additional \emph{edges} not from $G$, it is only natural to consider allowing additional \emph{vertices} not from $G$. Indeed, this question of \emph{Steiner DAG cover} was explicitly asked in both~\cite{AHW25} and~\cite{Fil25dags}. We initiate the study of Steiner DAG covers.

We focus, in particular, on the important graph classes of \emph{planar digraphs} and bounded \emph{treewidth} digraphs (i.e., the underlying undirected graph has bounded treewidth).  Our focus is motivated by the fact that some of the biggest success stories of tree covers are their applications to metrics with geometric or topological structure such as planar graphs and bounded treewidth graphs (in addition to e.g.~minor-free graphs, Euclidean metrics, and doubling metrics). Applications in these settings include, for example, distance oracle, spanners/emulators, low-hop emulator, distance labeling schemes, routing schemes, additive metric embeddings 
(e.g.~\cite{ADMSS95, CCLMST23Planar, chang2024shortcut, elkin2025spanning, FL22, GKR04, KLMS22}).

For these special classes of graphs, one can achieve much better tree covers than in general graphs. This raises the question:

\begin{quote} 
\emph{Can we obtain DAG covers for important special classes of graphs that are much better than those for general graphs (possibly with the use of Steiner vertices)?}
\end{quote}

Our main results are Steiner DAG covers for planar digraphs and bounded-treewidth digraphs that indeed have \emph{far better} parameters than the best-known results for general digraphs: we obtain DAG covers with only \emph{two} DAGs, and either 1 or $1+\varepsilon$ stretch. (Note that it is impossible to achieve only a single DAG if $G$ itself is not a DAG, even for preserving reachability.) 
In contrast, in general digraphs, there is a known \emph{lower bound} that achieving stretch 1 with $\tilde{O}(m)$ additional edges requires a \emph{polynomial} number of DAGs, and this is conjectured to extend to larger stretch~\cite{AHW25}. Additionally, stronger lower bounds are known for general digraphs when the number of added edges is $\tilde{O}(n^{2-\varepsilon})$: even for preserving reachability, $\Omega(n^{1-\varepsilon})$ DAGs are required~\cite{AHW25}.
Next, we will state our results formally.

\subsection{Our Results}

First we state the formal definition of a DAG cover.

\begin{definition}[DAG Cover]\label{def:DagCover}
	Given a weighted digraph $G=(V,E,w)$, a \EMPH{$(t,g,\mu)$-DAG cover} is a collection $\mathbb{D}$ of $g$ DAGs $D_1,\dots,D_g$ over $V$ such that: 
    \begin{itemize}        
        \item \emph{[Dominating:]} for every $u,v\in V$ and DAG $D_i\in\mathbb{D}$, $d_G(u,v)\le d_{D_i}(u,v)$.
        \item \emph{[Stretch:]} for every reachable pair $(u,v)$ it holds that $\min_{D_i\in\mathbb{D}} d_{D_i}(u,v)\le t\cdot d_G(u,v)$.
        \item \emph{[Sparsity:]} the total number of extra edges not in $G$ ($\left|\cup_i E(D_i)\setminus E(G)\right|$) is at most $\mu$.\footnote{Note that if all the DAGs in the cover are subgraphs of $G$, then $\mu=0$.\label{foot:subgraphsparsity}}
    \end{itemize}
    If the DAGs in $\mathbb{D}$ contain only vertices in $V$, we call it a \EMPH{non-Steiner DAG cover}. Otherwise if the DAGs contain vertices that are not in $V$, we call it a \EMPH{Steiner DAG cover}.
\end{definition}

\begin{table}[h!]
\footnotesize\sffamily\sansmath
\renewcommand{\arraystretch}{1.5}
\begin{tabular}{|c l l l l l|}
\hline
Family           & Steiner?                & Stretch $t$     & \# DAGs $g$                     & \# Added edges $\mu$                 & Ref. \\ \hline

\multirow{3}{*}{\shortstack{General\\ digraphs}} 
       & \ding{55} & $\tilde{O}(\log n)$ & $O(\log n)$ & $\tilde{O}(m)$ &  \cite{AHW25, Fil25dags} \\
         \cline{2-6} 
 & \ding{55} & finite & $\Omega(n^{1-\varepsilon})$ & $O(n^{2-\varepsilon})$ &  \cite{AHW25}  \\
		 \cline{2-6} 
	& \ding{55} &	$1$ & $\Omega(n^{1/6})$ & $O(m^{1+\eps})$ & \cite{AHW25} \\

		\hline

\multirow{3}{*}{\shortstack{Treewidth $\tw$\\ digraphs}} 
 & \ding{55} & $< 2$ & $\Omega(\log\frac{n^2}{\mu+1})$ & $\mu$                  & \Cref{thm:bidirectedLB}\\ \cline{2-6} 
 & \ding{55} & $1$          & $O(\log n)$                 & $\tilde{O}(n\cdot \tw)$                       & \Cref{thm:NonSteinerUB} \\ \cline{2-6} 
 & \ding{51} & $1$          & $2$                         & $\tilde{O}(n\cdot\tw)$       & \Cref{thm:treewidth}\\ \hline

Planar digraphs                          
 & \ding{51} & $1+\eps$     & $2$                         & $\tilde{O}(n\cdot \varepsilon^{-1}\cdot \log\Phi)$ &  \Cref{thm:planar}\\ \hline

\end{tabular}
\caption{
Summary of our results on DAG cover, in comparison to the best-known results for general digraphs. Rows with $\Omega$ in the \# DAGs column are lower bounds. Stretch ``finite'' means that the lower bound result holds even for preserving only reachability. The number of Steiner vertices in our constructions are asymptotically the same as the number of added edges.
	}\label{tab:results}
\end{table}

Our results are shown in \Cref{tab:results}. 
We show a stark separation between Steiner to non-Steiner DAG covers.
Specifically, in \Cref{thm:bidirectedLB} we show an example of a planar digraph with treewidth $1$ (the bidirected star) such that every non-Steiner DAG cover with stretch $t<2$, and $\mu$ additional edges must use \smash{$\Omega\big( \log(\frac{n^2}{\mu+1}) \big)$} DAGs.
In particular, for sub-quadratic $\mu\le n^{2-\Omega(1)}$, $\Omega(\log n)$ DAGs are required.
We then show that this lower bound (\Cref{thm:bidirectedLB}) is tight (up to second order terms).
Specifically, in \Cref{thm:NonSteinerUB} we show that any digraph with treewidth $\tw$ admits a non-Steiner $\left(1,O(\log n),\tilde{O}(n\cdot \tw)\right)$-DAG cover.
That is, for bounded treewidth digraphs, $O(\log n)$ DAGs without Steiner points are sufficient to preserve all distances exactly, while using near-linear number of additional edges.

\paragraph*{Treewidth.} In \Cref{thm:treewidth} we show that any digraph with treewidth $\tw$ admits a Steiner $(1,2,\tilde{O}(n\cdot\tw))$-DAG cover.
That is, we preserve all distances exactly, while using only $2$ DAGs, and near linear number of additional edges.
Compare this to non-Steiner DAG covers that require $\Omega(\log n)$ DAGs for such stretch and sparsity (\Cref{thm:bidirectedLB}).

\paragraph*{Planar.} Let $\Phi$ be the ratio between the maximum and minimum (finite) distances in $G$ (called aspect ratio), and let $\varepsilon\in (0,1)$. For planar digraphs, we show (\Cref{thm:planar}) that with $\tilde{O}(n\cdot \varepsilon^{-1}\cdot \log \Phi)$ added edges, nearly the best possible result is achievable: 2 DAGs and stretch $1+\varepsilon$.

\medskip
\noindent 
In all our results, the number of Steiner vertices is asymptotically the same as the number of added edges. See \Cref{sec:conclusion} for further discussion and open problems.

\subsection{Related Work}\label{sec:related}
\paragraph*{Trading off parameters in tree covers.} While~\cite{TZ05} seeked tree covers with small stretch, recently Bartal, Fandina, and Neiman~\cite{BFN22} studied the other side of the parameter regime for general (undirected) graphs, seeking a small number of trees. They obtained a tree cover with stretch $O(n^{1/k}\log^{1-1/k}n)$ and $k$ trees, for any integer $k\geq 1$. There is also a very recent lower bound in this parameter regime~\cite{chen2025lower}.

  \paragraph*{Ramsey tree covers.} In a \emph{Ramsey tree cover}, instead of having a collection of trees with a low-stretch estimate of each $d(u,v)$ in some tree, the goal is to have a collection of trees so that for each $u$ there is a low-stretch estimate of all $d(u,v)$ in a single tree. 
  That is, a collection of trees such that for every vertex, one of the trees is an approximate shortest path tree. See \cite{BLMN03, MN07,NT12,ACEFN20,FL21,Fil21,KLMS22,chen2025lower,elkin2025spanning}. 

  \paragraph*{Probabilistic tree and DAG embeddings.}
  Probabilistic tree embedding is an extensively studied and widely applicable primitive related to tree covers: the goal is to obtain a distribution over trees so that the expected value of each $d(u,v)$ has low stretch~\cite{Karp89, AKPW95, Bar96, bartal2, FRT04,Bartal04,KGR25}. See~\cite{AHW25} for a more detailed list of variants and applications of probabilistic tree embeddings. As an analog for directed graphs, probabilistic DAG embeddings have been recently studied~\cite{AHW25, Fil25dags}. 

\section{Preliminaries}
The \EMPH{$\tilde{O}$} notation hides polylogarithmic factors, that is $\tilde{O}(g)=O(g)\cdot\polylog(g)$. 
A \EMPH{digraph} $G=(V,E,w)$ consists of a set of vertices $V$, a directed set of edges $E\subseteq V\times V$ (with no self-loops), and a positive weight function $w:E\rightarrow\R_{>0}$. By scaling, we will assume w.l.o.g. that the minimum edge weight is $1$.
We will also use \EMPH{$V(G),E(G)$} to denote the set of vertices and edges of a digraph $G$.
For a subset $A\subseteq V$ of vertices, $\brick{G[A]}=(A,E_A=E\cap A\times A,w_{\upharpoonright E_{A}})$ denotes the subgraph \EMPH{induced} by $A$, i.e., the graph with vertex set $A$ where we keep all edges between vertices of $A$.
Given a digraph $G$, \EMPH{$s\rightsquigarrow_{G}t$} denotes that there is a path from $s$ to $t$ in $G$. Similarly, we will use $s\not\rightsquigarrow_{G}t$ to denote that there is no path from $s$ to $t$ in $G$.

Given a path $\pi=(v_0,v_1,\dots,v_k)$, $\brick{\pi[v_i,v_j]}=(v_i,v_{i+1},\dots,v_j)$ denotes the subpath from $v_i$ to $v_j$. Similarly, we use $\pi(v_i,v_j)=(v_{i+1},v_{i+2},\dots,v_{j-1})$ to denote the subpath without endpoints (we might also use $\pi[v_i,v_j)$ and $\pi(v_i,v_j]$ to include only one endpoint).
The notation \EMPH{$d_{G}$} denotes the shortest path (quasi)metric in $G$, i.e., $d_G(u,v)$ is the minimal weight of a path from $u$ to $v$. 
If there is no such path ($u\not\rightsquigarrow_{G}v$) then we set $d_G(u,v)=\infty$.
The \EMPH{aspect ratio} is 
$\Phi=\frac{\max\left\{d(u,v)\mid u\rightsquigarrow_{G}v\right\}}{\min\left\{d(u,v)\mid u\rightsquigarrow_{G}v\right\}}$, the ratio between the maximum and minimum (finite) distances in $G$. 

A \EMPH{directed acyclic graph (DAG)} is a digraph not containing any directed cycle. 
In particular, the set of strongly connected components (SCCs) is the set of singleton vertices.
Given a DAG $D=(V,E)$, a \EMPH{topological order} is a total order $<_{\rm TO}$ over $V$ such that for every edge $(u,v)\in E$, $u<_{\rm TO}v$.
Topological order is not necessarily unique.
Consider two DAGs $D=(V,E)$ and $D'=(V',E')$ with respective topological orders $<_D$ and $<_{D'}$ respectively.
We say that the two orders \EMPH{agree} if for every $u,v\in V\cap V'$, $u<_Dv\iff u<_{D'}v$.
Note that if the orders of two DAGs $D,D'$ agree, then $D\cup D'=(V\cup V',E\cup E')$ is also a DAG, and the topological order of $D\cup D'$ restricted to $D$ (resp.\ $D'$) is identical to $<_D$ (resp.\ $<_{D'}$).

\paragraph*{Treewidth.} A \EMPH{tree decomposition} of an undirected graph $G = (V,E)$ is a  tree $\mathcal{T}$ where each node $x\in \mathcal{T}$ is associated with a subset $S_x$ of $V$, called a \EMPH{bag}, such that: (i) $\cup_{x\in V(\mathcal{T})} S_x = V$, (ii) for every edge $(u,v) \in E$, there exists a bag $S_x$ for some $x\in V(\mathcal{T})$ such that $\{u,v\}\subseteq S$, and (iii) for every $u\in V$, the bags containing $u$ induces a connected subtree of $\mathcal{T}$. The \emph{width} of $\mathcal{T}$ is $\max_{x\in V(\mathcal{T})}\{|S_x|\}$-1. The \EMPH{treewidth} of $G$ is the minimum width among all possible tree decompositions of $G$. 
We say that a digraph $G$ has treewidth $\tw$ if it's undirected counterpart (where we keep all edges and ignore directions) has treewidth $\tw$.\\
A \EMPH{path decomposition} of $G$ is a special kind of tree decomposition where the underlying tree is a path.
The \EMPH{pathwidth} of $G$ is the minimal width of a path decomposition of $G$.

\section{Non-Steiner DAG cover for digraphs with bounded treewidth}
This section is devoted to non-Stiener DAG cover for digraphs with bounded treewidth. 
We begin with a lower bound (\Cref{thm:bidirectedLB}), which provides an example of a planar digraph with treewidth $1$ (bidirected star), showing that any DAG cover with a subquadratic number of additional edges $\mu=O(n^{2-\delta})$, requires $\Omega(\log n)$ DAGs.
Afterwards, in \Cref{thm:NonStienerUB} we show that \Cref{thm:bidirectedLB} is tight (up to second order terms) by proving that every digraph with bounded treewidth $\tw$ admits a DAG cover with stretch $1$, $2$ DAGs, and near-linear number of additional edges.

\begin{theorem}[Non-Steiner LB]\label{thm:bidirectedLB}
    There is a digraph $G$ with treewidth $1$ such that for every $t<2$, and $\mu\in[0,n^2]$, if $G$ admits a $(t,g,\mu)$-DAG cover without Steiner points, then 
    $g \ge \log\left(\frac{(n-1)^2}{2\mu+n-1}+1\right) = \Omega\left(\log(\frac{n^2}{\mu+1})\right)$.
\end{theorem}

\begin{wrapfigure}{r}{0.15\textwidth}
	\begin{center}
		\vspace{-20pt}
		\includegraphics[width=0.9\textwidth]{ 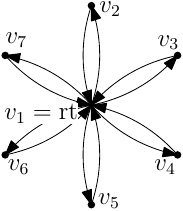}
		\vspace{-10pt}
	\end{center}
	\vspace{-10pt}
\end{wrapfigure}
\par\noindent\textit{Proof.}\hspace{0.2em}
The bidirected star $\cS_n$ is the unweighted digraph consisting of a root vertex $\rt$ and $n$ ``leaf'' vertices $v_2,\dots,v_n$, where there are bidirected edges between any leaf vertex $v_i$ and the root $\rt$ (formally the edges are $\left\{(\rt,v_i)\right\}_{i=2}^{n}\cup\left\{(v_i,\rt)\right\}_{i=2}^{n}$). 
See illustration on the right.

Consider a $(t,g,\mu)$-DAG cover $\mathbb{D} = \{D_1, D_2,\dots,D_g\}$ of $\cS_n$.
We begin by making an observation about the shortest path $P$ in $\mathbb{D}$ between two leaf vertices $v_i$ and $v_j$ with $i\neq j$.
Consider a DAG $D$, where
the shortest 
$v_i$-$v_j$ path $P$ goes through some other leaf vertex $v_h$. 
Then, by the dominating property of the DAG cover, 
\[d_D(v_i,v_j) = d_D(v_i,v_h) + d_D(v_h, v_j) \ge d_G(v_i, v_h) + d_G(v_h, v_j) = 2 + 2 > t\cdot d_G(v_i, v_j)~,\]
thus the distortion guarantee is not fulfilled in $D$.
We conclude that one of the following must hold: either there is a DAG in the cover with a direct edge between $v_i$ and $v_j$, or there is a DAG in the cover containing the path $(v_i,\rt,v_j)$. 

To proceed, we will define a binary string of length $g$ that we will call a \EMPH{codeword} for every vertex, and show that there must be a large number of different codewords. These codewords will be derived from the above observation.
Formally, for each vertex $v_i$, we will construct a codeword $B_i = b_i^1b_i^2\dots b_i^g$ of length $g$ where each bit $b_i^k$ is $1$ if the DAG $D_k$ has an edge from $v_i$ to $\rt$, and $0$ otherwise. 
Consider the set of unordered pairs: 
\[ 
Q \coloneqq \left\{~\{v_i, v_j\} : \begin{matrix*}[l]
\text{there is no edge in any DAG of $\mathbb{D}$ between} \\ \text{$v_i$ and $v_j$ (in either direction) for $i\neq j$} 
\end{matrix*}
\right\}.
\]
$Q$ contains at least ${n-1 \choose 2} - \mu$ pairs, as any edge in the DAG cover $\mathbb{D}$ (not in $\cS_n$) can remove only a single edge from $Q$.
Notice that for every pair $\{v_i,v_j\}\in Q$, there must be some DAG $D_q\in \mathbb{D}$ containing the path $(v_i,\rt,v_j)$. As $D_q$ is a DAG, it will not contain the edge $(v_j,\rt)$. In particular $b_i^q=1$, $b_j^q=0$.
We conclude that for every pair $\{v_i,v_j\}\in Q$,
the codewords $B_i$ and $B_j$ must differ in some bit.

Consider the undirected graph $G'$ with vertices $\{v_2,\dots,v_n\}$, and the edge set $Q$. 
By Tur\'an's theorem, the graph $G'$ contains a clique of size $r$ if 
\[ |Q|\ge{n-1\choose 2} - \mu> \left(1-\frac{1}{r-1}\right)\frac{(n-1)^2}{2}~. \]
It follows that $G'$ contains a clique  of size at least
$r=\frac{(n-1)^2}{2\mu+n-1}+1$.
Each of the vertices in this clique  must have a different codeword. Since there are at most $2^{g}$ different possible codewords of length $g$, we conclude that $2^{g} \ge r$, and thus $g \ge \log\left(\frac{(n-1)^2}{2\mu+n-1}+1\right) = \Omega\left(\log(\frac{n^2}{\mu+1})\right)$.
\qed
\vspace{15pt}

From \Cref{thm:treewidth} it follows that the bidirected star (the lower bound example of \Cref{thm:bidirectedLB}) admits a Steiner $(1,2,\tilde{O}(n))$-DAG cover, a huge difference!
See \Cref{fig:doubleStar} for an illustration of this DAG cover.
\Cref{thm:NonSteinerUB} below shows that our lower bound from \Cref{thm:bidirectedLB} is tight.
That is, there is a DAG cover preserving all distances exactly with $O(\log n)$ DAGs, while only using near-linear number of edges.

\begin{theorem}
\label{thm:NonSteinerUB}
Every $n$-vertex digraph admits a 
$\left(1,O(\log n),O(n\cdot\tw\cdot\log^2 n)\right)$-DAG cover,
where $\tw$ denotes the treewidth of the graph.
\end{theorem}
\begin{proof}
    \sloppy
    The proof is by induction on the number of vertices $n$. Specifically, we will assume that for every digraph with treewidth $\tw$ and $n'<n$ vertices admits an $\left(1,2\lceil\log n'\rceil,n'\cdot2(\tw+1)\cdot \lceil\log n'\rceil^2\right)$ DAG cover. The base case is trivial.
    Consider a  digraph $G=(V,E,w)$ with treewidth $\tw$ and $n$ vertices.
    Let $B=\{x_1,\dots,x_q\}\subseteq V$ be a separator bag. That is, $B$ is a set of $|B|\le\tw+1$ vertices, such that every connected component of $G\setminus B$ (in the undirected sense), contains at most $\frac n2$ vertices. 
    Consider $G\setminus B$, and let $C_1,\dots,C_\ell$ be the connected components in the undirected sense. 
    For every $j$, let $\mathbb{D}_j=\{D^j_1,D^j_2,\dots,D^j_{2\lceil\log |C_j|\rceil}\}$ be the DAG cover of $G[C_j]$ promised by induction.
    We also let $D^j_{2\lceil\log |C_j|\rceil+1},\dots,D^j_{2\lceil\log n\rceil}$ to be empty digraphs over $C_j$.

    Let $S^1,\dots,S^g\in\{0,1\}^{2\lceil\log n\rceil}$ be codewords such that for every two codewords $S^i,S^j$ there are two bits $\alpha,\beta$ where $S^i_\alpha=0,S^i_\beta=1,S^j_\alpha=1,S^j_\beta=0$. Such codewords can be created e.g. by concatenating the binary representations of the numbers $i-1$ and $n-(i-1)$ for each $i$.
    We create the DAG cover $\mathbb{D}=\{D_1,D_2,\dots,D_{2\lceil\log n\rceil}\}$. The DAG $D_i$ is constructed as follows:
    \begin{enumerate}
        \item $D_i$ will contain all the edges in $D^1_i,D^2_i,\dots,D^\ell_i$.
        \item For every $j\in[1,\ell]$, if $S^j_i=1$ add the edges $C_j\times B$, that is all possible outgoing edges from $C_j$ to $B$. Otherwise ($S^j_i=0$), add the edges $B\times C_j$, that is all possible outgoing edges from $B$ to $C_j$. 
        \item If $i=1$, add the edges $\left\{(x_\alpha,x_\beta)\mid 1\le\alpha<\beta\le q\right\}$, that is all possible internal edges in the bag $B$ that respect the order $(x_1,\dots,x_q)$.
        \item If $i=2$, add the edges $\left\{(x_\beta,x_\alpha)\mid 1\le\alpha<\beta\le q\right\}$, that is all possible internal edges in the bag $B$ that respect the order $(x_q,\dots,x_1)$.
    \end{enumerate}
    The weight of every added edge $(x,y)$ will be $d_G(x,y)$.

  First, we observe that each $D_i$ is a DAG. 
    Indeed, we can define a topological order $<_i$ over the vertices $V$ such that all the edges respecting the order. The order inside each connected component $C_j$ is defined with respect to  a topological order of $D^j_i$.
    All the components $C_j$ with $S^j_i=1$ will appear before $B$ vertices in the order, while all other components (with $S^j_i=0$) appear after $B$. The internal order on $B$ in $D_1$ is $(x_1,\dots,x_q)$, in $D_2$ is $(x_q,\dots,x_1)$, and for $i\ge 3$ is arbitrary. We conclude that we obtained a DAG cover.

    Next we bound the total number of edges.
    In steps 2-4 we added $(n-|B|)\cdot |B|$ edges for each DAG. 
    We also added $2\cdot{|B|\choose 2}\le |B|^2$  edges (in total) between the vertices of $B$. 
    Thus in total $|B|^{2}+(n-|B|)\cdot|B|\cdot2\lceil\log n\rceil\le n\cdot|B|\cdot2\lceil\log n\rceil$ edges.
    Using the induction hypothesis, the total number of edges added due to previously created DAG covers is $\sum_{i=1}^{\ell}|C_{i}|\cdot(\tw+1)\cdot2\lceil\log|C_{i}|\rceil^{2}\le n\cdot(\tw+1)\cdot2\lceil\log\frac{n}{2}\rceil^{2}$.
    We conclude that the total number of edges in the DAG cover is 
    \[
    n\cdot|B|\cdot2\lceil\log n\rceil+n\cdot(\tw+1)\cdot2\lceil\log\frac{n}{2}\rceil^{2}\le n\cdot(\tw+1)\cdot2\lceil\log n\rceil^{2}~,
    \]
    as required.

    Finally we argue that all the distances are exactly preserved. Consider a pair of vertices $x,y\in V$. Let  $x=x_\alpha$ and $y=x_\beta$, where w.l.o.g.\ $\alpha<\beta$. We continue by case analysis.
    \begin{itemize}
        \item $x,y\in B$.  The edge $(x,y)$ was added to $D_1$, and thus $d_{D_1}(x,y)=d_G(x,y)$.

        \item $x\in B,y\notin B$, then there is some $j$ such that $y\in C_j$, and some index $\alpha$ such that $S^j_\alpha=0$. In $D_\alpha$ we added the edge $(x,y)$.

        \item $y\in B,x\notin B$, then there is some $j$ such that $x\in C_j$, and some index $\alpha$ such that $S^j_\alpha=1$. In $D_\alpha$ we added the edge $(x,y)$.

        \item $x,y\notin B$, and the shortest $x$-$y$ path does not intersect $B$. By the definition of treewidth, there is some $j$ such that the shortest $x$-$y$ path lies in $C_j$ (and in particular, $x,y\in C_j$). Using the induction hypothesis there is some DAG $D^j_\alpha$ such that $d_{D_\alpha}(x,y)\le d_{D^j_\alpha}(x,y)=d_G(x,y)$.

        \item $x,y\notin B$, and the shortest $x$-$y$ path intersects $B$.
        Let $z\in B$ be a vertex such that the shortest $x$-$y$ path goes through $z$.
        Let $i$ be an index such that $S^j_i=1$ and $S^{j'}_i=0$.
        In $D_i$ we added the edges $(x,z),(z,y)$. It thus holds that 
        $d_{D_i}(x,y)\le d_{G}(x,z)+d_{G}(z,y)=d_{G}(x,y)$.\qedhere
    \end{itemize}
\end{proof}
\begin{remark}
    Note that the treewidth of each of the DAGs in the cover created in \Cref{thm:NonSteinerUB} is $O(\tw\cdot\log n)$.
\end{remark}

\section{Steiner DAG cover for Digraphs with Bounded Treewidth}
In this section we provide a DAG cover for bounded treewidth graphs. Furthermore, we show that the structure of the DAGs we construct preserves some structure of the original graph as it has bounded pathwidth.
\begin{theorem}\label{thm:treewidth}
    Every $n$-vertex digraph with treewidth $\tw$ admits a $\left(1,2,O(n\cdot\tw\cdot\log n)\right)$-Steiner DAG cover. 
    Furthermore, both of the DAGs in the cover have pathwidth $O(\tw\cdot\log n)$.
\end{theorem}
\begin{proof}
    Let $G=(V,E,w)$ be an $n$-vertex digraph with treewidth $\tw$. 
    We will begin by disregarding the pathwidth requirement. Later we will adapt the construction accordingly.
    Fix an arbitrary permutation $\sigma=(v_1,\dots,v_n)$ of the vertices in $V$. Let $\sigma^{-1}=(v_n,\dots,v_1)$ be the reversed permutation.
    The work horse of our construction is a gadget that allows us to preserve all the shortest paths respecting $\sigma$ that go through a certain separator vertex. We will then use this gadget on all separator vertices, and finally use it recursively to obtain our first DAG $D_\sigma$. 
    The second DAG $D_{\sigma^{-1}}$ will be constructed in the same way with respect to  the reversed permutation.

    \begin{lemma}\label{lem:vertexGadget}
        Given an $n$ vertex digraph $G=(V,E)$, order $\sigma$ over $V$, and a vertex $x\in V$, there is a DAG $D_x$ containing $V$ such that
        \begin{itemize}
            \item The induced topological order of $D_x$ on $V$ is $\sigma$.
            \item For every $u,v\in V$ such that $u<_\sigma v$, $d_G(u,v)\le d_{D_x}(u,v)\le d_{G}(u,x)+d_{G}(x,v)$.
        \end{itemize}
        $D_x$ contains $n$ Steiner points, and at most $3n$ edges.
    \end{lemma}
    \begin{proof}
        Let $\sigma=(v_1,\dots,v_n)$ be the designated order over $V$.
        For every vertex $v_i$, add an auxiliary Steiner point $u_i$. 
        For every $i\in[1,n]$ we add to $D_x$ the following edges:
        \begin{itemize}
            \item If $i<n$, the edge $(u_i,u_{i+1})$ of weight $0$.
            \item If  $x\rightsquigarrow_{G} v_i$, the edge $(u_i,v_i)$ of weight $d_G(x,v_i)$.
            \item If $i<n$  and $v_i\rightsquigarrow_{G} x$, the edge $(v_i,u_{i+1})$ of weight $d_G(v_i,x)$.
        \end{itemize}
        See \Cref{fig:doubleStar} for an illustration. 
        Note that in total we added at most $3n$ edges.
        $D_x$ is indeed a DAG, where the topological order is $u_1,v_1,u_2,v_2,\dots,u_n,v_n$ (note that all the edges respect this order). Accordingly, the induced topological order over $V$ is $\sigma$.
        Finally, for every $i<j$ such that  $v_i\rightsquigarrow_{G} x$ and $x\rightsquigarrow_{G} v_j$ it holds that the only outgoing edge from $v_i$ is towards $u_{i+1}$, while the only ingoing edge to $v_j$ is from $u_j$. 
        It follows that
        \[
        d_{D_{x}}(v_{i},v_{j})=d_{D_{x}}(v_{i},u_{i+1})+d_{D_{x}}(u_{i+1},u_{j})+d_{D_{x}}(u_{j},v_{j})=d_{G}(v_{i},x)+0+d_{G}(x,v_{j})~,
        \]
        and in particular $d_{D_{x}}(v_{i},v_{j})\ge d_{G}(v_{i},v_{j})$.
        Note that if $v_i\not\rightsquigarrow_{G} x$ or $x\not\rightsquigarrow_{G} v_j$, then $d_{G}(v_{i},x)+d_{G}(x,v_{j})=\infty$, and the lemma holds trivially.
    \end{proof}

\begin{figure}[t]
\includegraphics[width=0.98\textwidth]{ 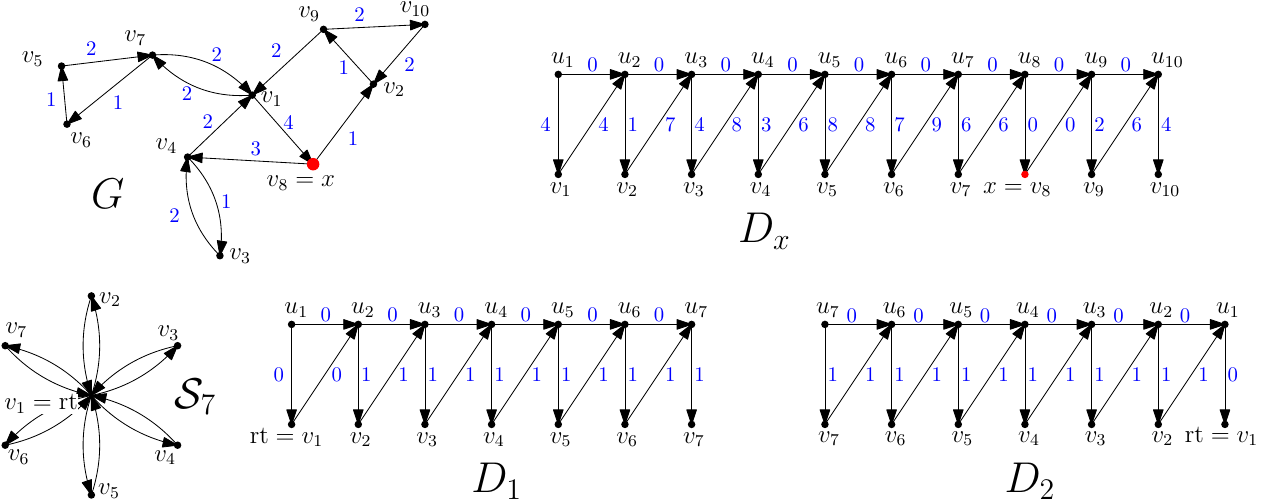}
    \caption{
    On top: Illustration of the DAG $D_x$ created using \Cref{lem:vertexGadget} for the digraph $G$ (on the left) with respect to  the designated vertex $x$ ($v_8$), and the order $(v_1,v_2,\dots,v_{10})$.
    The DAG $D_x$ contains an auxiliary vertex $u_i$ for each $v_i$. We add a $0$ weight edge between any two consecutive $u_i,u_{i+1}$. In addition for every $i$ there is an edge from $u_i$ to $v_i$ of weight $d_G(x,v_i)$, and from $v_i$ to $u_{i+1}$ of weight $d_G(v_i,x)$.
    For every $i<j$, $d_{D_x}(v_i,v_j)\le d_{D_x}(v_i,u_{i+1})+d_{D_x}(u_{i+1},u_j)+d_{D_x}(u_j,v_j)=d_G(v_i,x)+0+d_G(x,v_j)$. That is, the shortest $v_i$-$v_j$ path that goes through $x$.
    \\ 
    On bottom: Illustration of the bidirected star $\cS_7$ (used in the proof of \Cref{thm:bidirectedLB}), and its DAG cover. Here we take the order \smash{$\sigma=(v_1,\dots,v_7)$} where the root vertex $v_1=\rt$ is first.
    \smash{$D_1=D_{\rt}$} is created using \Cref{lem:vertexGadget} with respect to  vertex $x$ and order $\sigma$.
    $D_2$ is created in the same way, but with respect to  order $\sigma^{-1}$.
    As all shortest paths in $\cS_7$ go through $\rt$, $D_1$ and $D_2$ preserve all the distances exactly.
    }
\label{fig:doubleStar}
\end{figure}

    We will prove \Cref{thm:treewidth} by induction on the number of  vertices $n$. Suppose that every digraph $G=(V,E)$ with $k<n$ vertices, treewidth $\tw$, and for every permutation $\sigma$ of $G$ vertices, there is a dominating DAG $D_\sigma$ with $3\cdot k\cdot(\tw+1)\cdot\log k$ edges such that the induced topological order on $V$ is $\sigma$, and for every $u<_\sigma v$, $d_{D_\sigma}(u,v)=d_G(u,v)$.

    Next, consider a digraph $G=(V,E)$ with $n$ vertices, treewidth $\tw$, and an arbitrary permutation $\sigma$ of $G$ vertices.
    Let $B=\{x_1,\dots,x_q\}\subseteq V$ be a separator bag. That is, $B$ is a set of $|B|\le\tw+1$ vertices, such that every connected component of $G\setminus B$ (in the undirected sense), contains at most $\frac n2$ vertices. 
    For every separator vertex $x_i\in B$, let $D_{x_i}$ be the DAG returned by \Cref{lem:vertexGadget}, with respect to  the permutation $\sigma$ and the vertex $x_i$.
    Consider $G\setminus B$, and let $C_1,\dots,C_g$ be the connected components in the undirected sense. 
    Let $\sigma_j$ be the order $\sigma$ induced on the vertices $C_j$.
    Using the induction hypothesis, let $D_{\sigma_j}$ be the DAG created for $G[C_j]$ with respect to  the order $\sigma_j$. We set $D_\sigma$ to be the digraph created by taking the union of all the DAGs $D_{x_1},\dots,D_{x_q},D_{\sigma_1},\dots,D_{\sigma_g}$.
    The Steiner points used in each DAG will be disjoint.
    Observe that $D_\sigma$ is a DAG with topological order $\sigma$. This follows as all the DAGs in the union respect the order $\sigma$.
    Note also that all the distance in $D_{\sigma}$ are dominating $G$ because the Steiner points are disjoint (and a union of dominating digraphs with disjoint Steiner points is dominating).

    Next we bound the size of $D_\sigma$.
    Using the induction hypothesis, it holds that $D_{\sigma_j}$ contains at most $|C_{j}|\cdot 3\cdot(\tw+1)\cdot\log|C_{j}|\le|C_{j}|\cdot 3\cdot(\tw+1)\cdot\log\frac{n}{2}$ edges.
    By \Cref{lem:vertexGadget} and the fact $\sum_{j=1}^{g}|C_{j}|< n$, the number of vertices in $D_\sigma$ is bounded by
    \begin{align*}
        |D_{\sigma}| & \le\sum_{i=1}^{q}|D_{x_{i}}|+\sum_{j=1}^{g}|D_{\sigma_{j}}|\\
         & \le(\tw+1)\cdot3n+3\cdot(\tw+1)\cdot\log\frac{n}{2}\cdot\sum_{j=1}^{g}|C_{j}|\\
         & \le(\tw+1)\cdot3n\cdot\left(1+\log\frac{n}{2}\right)=(\tw+1)\cdot3n\cdot\log n.
    \end{align*}

    \noindent
    Next, we argue that distances are preserved. Consider $u,v$ such that $u<_\sigma v$. If $u\not\rightsquigarrow_{G} v$, then there is nothing to prove. We will thus suppose that $u\rightsquigarrow_{G} v$. Let $P$ be the shortest $u$-$v$ path in $G$.
    Suppose first that there is some vertex $x_i\in P\cap B$. In this case, 
    \[d_{D_\sigma}(u,v)\le d_{D_{x_i}}(u,v)\le d_G(u,x)+d_G(x,v)=d_G(u,v).\]
    We will thus assume that $P$ is disjoint from the separator $B$.
    It follows that $P$ vertices are fully contained in some (undirected) connected component $C_j$ of $G\setminus B$, and in particular $d_{G[C_j]}(u,v)=d_{G}(u,v)$. Using the inductive hypothesis, it follows that $$d_{D_\sigma}(u,v)\le d_{D_{\sigma_j}}(u,v)=d_{G[C_j]}(u,v)=d_{G}(u,v)~.$$

    Finally, we construct a similar DAG $D_{\sigma^{-1}}$ with respect to the order $\sigma^{-1}$. As for every $u,v\in V$ either $u<_{\sigma} v$ or $u<_{\sigma^{-1}} v$, it follows that 
    $\min\big\{d_{D_\sigma}(u,v),d_{D_{\sigma^{-1}}}(u,v)\big\} = d_G(u,v)$.

    The proof that it is possible to create a DAG cover so that both DAGs have pathwidth $O(\tw\cdot\log n)$ is deferred to Appendix~\ref{sec:pathwidth}.
\end{proof}

\section{Planar Digraphs}

This section is dedicated to proving the following theorem:
\begin{theorem}\label{thm:planar}
Given an $n$-vertex planar digraph $G=(V,E,w)$, and $\eps\in(0,1)$, $G$ admits a $\left(1+\eps,2,O(n\cdot\eps^{-1}\cdot\log^2 n\cdot\log\Phi)\right)$-Steiner DAG cover.
\end{theorem}

\paragraph*{Shortest path separators.}
Thorup \cite{Tho04} introduced shortest path separators for planar digraphs and used them to construct a compact labeling scheme to approximate distances.
The main tool in Thorup's construction is a collection of paths such that each vertex is associated with small number of paths, and every pairwise distance can be well approximated by rerouting through the associated paths.
We summarize the exact details below.
\begin{lemma}[Path cover for planar digraphs \cite{Tho04}]
\label{lem:distancelabel}
    Let $G = (V,E, w)$ be a planar digraph with aspect ratio $\Phi$, and let $\eps > 0$. There is a set of dipaths $\cP$ in $G$, such that
    \begin{itemize}
        \item every vertex $v$ in $V$ is associated with a subset \EMPH{$\cP[v]$} of $O(\log n \cdot \log \Phi)$ paths of $\cP$;
        \item every pair $(v, P)$ of a vertex $v \in V$ and a path $P \in \cP[v]$ is associated with a set of $O(\eps^{-1})$ vertices \EMPH{$C(v, P)$} $\subseteq P$, called \EMPH{$\eps$-covering set},  with the property that
    for any pair of vertices $u$ and $v$ in $V$ with $u \rightsquigarrow_G v$, there exists some path $P \in \cP[u] \cap \cP[v]$ and there exist vertices $u' \in C(u,P)$ and $v' \in C(v, P)$ such that \[d_G(u,u') + d_P(u',v') + d_G(v', v) \le (1+\eps) d_G(u,v).\]
    \end{itemize}
   
\end{lemma}

\Cref{lem:distancelabel} follows from Thorup \cite{Tho04} implicitly. Specifically, it follows by combining \cite{Tho04} Lemmas 3.2, 3.4, and 3.5.
Similar statements have also been made in \cite{MS18,LM19}.

\paragraph*{Steiner DAG cover construction.}
As in the bounded treewidth case (\Cref{thm:treewidth}), we 
fix an arbitrary permutation $\sigma=(v_1,\ldots,v_n)$ of $V$.
We construct one DAG $D_\sigma$ that preserves all distances for pairs $u,v$ such that $u<_{\sigma}v$ within a multiplicative error of $1+\eps$ (the error arises from \Cref{lem:distancelabel}).
We construct another DAG $D_{\sigma^{-1}}$, defined symmetrically, that preserves all remaining distances.
To obtain $D_\sigma$ (and $D_{\sigma^{-1}}$),
similarly to the bounded treewidth case, we can take the union of DAGs constructed using \Cref{lem:vertexGadget}. 
Suppose first that for each vertex $x\in V$, we use \Cref{lem:vertexGadget} to construct a DAG $D_x$ with topological order $\sigma$, such that for all pairs $u <_\sigma v$ it holds that $d_{D_x}(u,v)=d_G(u,x)+d_G(x,v)$. 
We will then define our DAG as a union of all these DAGs $D_{\rm all}=\cup_{x\in V}D_x$.
Following our proof and construction of bounded treewidth digraphs (\Cref{thm:treewidth}),  $D_{\rm all}$ is a DAG (as it is a union of DAGs, with intersection only over $V$, where the topological order with respect to  $V$ is $\sigma$ in all the DAGs),
and $D_{\rm all}$ preserve exactly the distance between all $(u,v)$ pairs such that $u <_\sigma v$.
Unfortunately, $D_{\rm all}$ has quadratic size.

Instead, for each vertex $x\in V$, we will define a set of associated centers $\brick{A_x}\subseteq V$.
We then apply \Cref{lem:vertexGadget} to construct $D_x$ only over $A_x$ (and not over $V$). 
We then define $D_\sigma=\cup_{x\in V}D_x$. 
Clearly $D_\sigma$ is a dominating DAG.
Let $\brick{X_v}=\{x\in V\mid v\in A_x\}$ be all the centers of the DAGs to which $v$ joins.
Following \Cref{lem:vertexGadget}, the number of edges in $D_\sigma$ will be $\sum_{x\in V}O(|A_x|)=\sum_{v\in V}O(|X_v|)$.
Our goal is thus to choose the sets $\{X_v\}_{v\in V}$ so that their total size is bounded, while all pairwise distances are preserved. That is for all $u,v\in V$
$\min_{x\in X_u\cap X_v}\{d_G(u,x)+d_G(x,v)\}\le(1+\eps)\cdot d_G(u,v)$. 

Recall that in the treewidth case (\Cref{thm:treewidth}), $X_v$ consist of the vertices in the separator bag $B$, and then recursively of all the vertices in associated  separator bags in the weakly connected component of $v$ in $G\setminus B$. That ensured that $|X_v|=O(\tw\cdot\log n)$, and that all the distances are preserved exactly.
For the planar case, $X_v$ will be chosen with respect to  the associated paths $\cP[v]$, and the $\eps$-covering sets $\{C(v,P)\}_{P\in\cP[v]}$ (given by \Cref{lem:distancelabel}).

\begin{figure}[t]
	\begin{center}
		\includegraphics[width=0.8\textwidth]{ 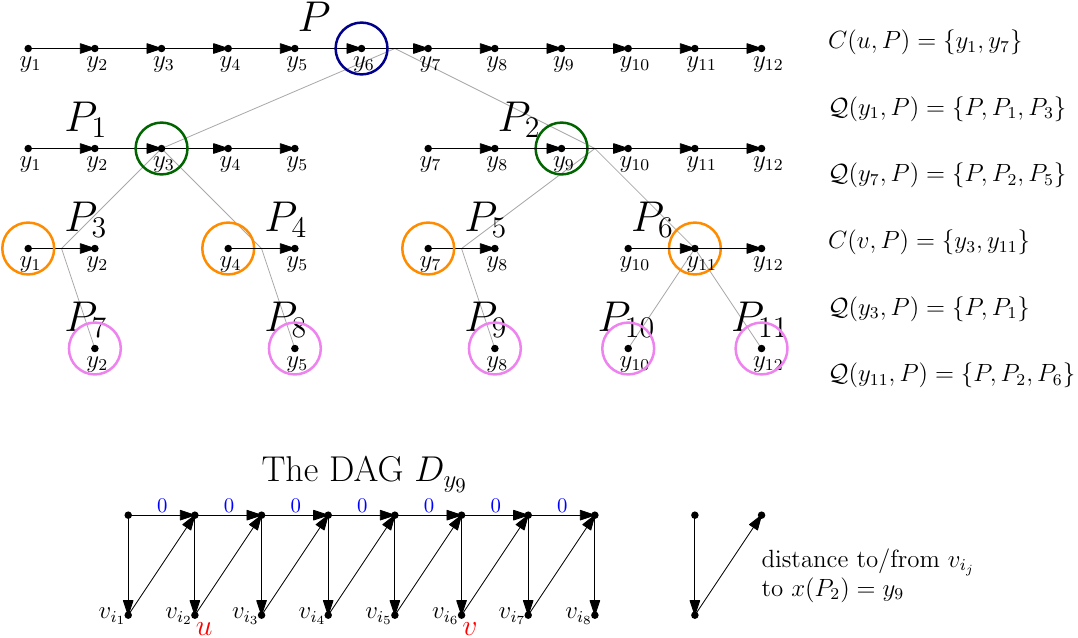}
		\caption{
        On top: Illustration of the path $P\in\cP$ and how it is broken to subpaths by deleting centroid vertices. 
        The hierarchical tree of the subpaths is drawn using a thin gray line.
        For each vertex $y_i\in P$, $\cQ(y_i,P)$ contains all the subpaths containing $y_i$.
        Consider two vertices $u<_\sigma v$ such that $P\in \cP[u]\cap \cP[v]$, and there are vertices $u'\in C(u,p)$, $v'\in C(v,p)$ where $d_G(u,u') + d_P(u',v') + d_G(v', v) \le (1+\eps) d_G(u,v)$.
        In the illustration, $u'=y_7$ and $v'=y_{11}$.
        The subpath $P_2$ contains both $y_7,y_{11}$, and thus both $X_u$ and $X_v$ will contain $x(P_2)=y_9$, the centroid of $P_2$.
        \\
        On bottom: Illustrated the DAG $D_{y_9}$ over $A_{y_9}$ that contains both $u,v$ (constructed using \Cref{lem:vertexGadget}).
        As $u<_\sigma v$, $d_{D_{y_9}}\le d_G(u,y_9)+d_G(y_9,v)\le (1+\eps)\cdot d_{G}(u,v)$.}
    \label{fig:PathHirarchy}
	\end{center}
    \vspace{-10pt}
\end{figure}

Consider a separator path $P\in\cP$.
We construct a hierarchical decomposition of $P$ into $O(\log |V(P)|)$ levels. The first level consists of $P$ itself.
Let $x(P)$ be the \EMPH{centroid} of $P$. 
That is, a vertex whose number of predecessors and successors along $P$ differ by at most one. 
$P$ will have two child paths: the prefix of $P$ up to (and not including) $x(P)$, and the suffix of $P$ from (and not including) $x(P)$.
We continue recursively on each of these two subpaths. See \Cref{fig:PathHirarchy} for an illustration.
Overall, we obtain a binary tree of paths, where for each path $P'$ in this tree with children $P_1',P_2'$ it holds that $P'=P_1'\circ x(P')\circ P_2'$.
The leaves in this tree will be singleton subpaths (with a single vertex).
Each vertex $y\in P$ will be associated with a set \EMPH{$\cQ(y,P)$} of subpaths containing it, where $|\cQ(y,P)|=O(\log |P|)$.
For a vertex $v\in V$, the set $X_v$ of DAG centers to which $v$ joins is defined as follows:
{
\renewcommand{\labelitemii}{\labelitemi}
\renewcommand{\labelitemiii}{\labelitemi}
\renewcommand{\labelitemiv}{\labelitemi}

\begin{itemize}
    \item For every path $P\in\cP[v]$,
    \begin{itemize}
        \item for every vertex $y\in C(v,P)$ from the $\eps$-covering set,
        \begin{itemize}
            \item for every subpath $P'\in \cQ(y,P)$,
            \begin{itemize}
                \item add the centroid $x(P')$ to $X_v$.
            \end{itemize}
        \end{itemize}
    \end{itemize}
\end{itemize}
}

\smallskip
Clearly, as $|\cP[v]|=O(\log\Phi\cdot \log n)$, $C(v,P)=O(\eps^{-1})$ (for every $P\in C(v,P)$ ), and $\cQ(y,P)=O(\log n)$ (for every $y\in C(v,P)$ ), it holds that $|X_v|=O(\eps^{-1}\cdot\log\Phi\cdot \log^2 n)$.
The definition of the sets $\{A_x\}_{x\in V}$ follows. 
For every vertex $x\in V$, we apply \Cref{lem:vertexGadget} with respect to  the vertex $x\in V$, the subset $A_x$, and the topological order $\sigma$, to obtain the DAG \EMPH{$D_x$}. We then define  $\brick{D_\sigma}=\cup_{x\in V}D_x$.
Following the lines of \Cref{thm:treewidth}, $D_\sigma$ is a DAG.
The number of edges in $D_\sigma$ is bounded by $\sum_{x\in V}O(|A_x|)=\sum_{v\in V}O(|X_v|)=O(n\cdot\eps^{-1}\cdot\log\Phi\cdot \log^2 n)$.
It remains to prove the stretch guarantee.

Consider a pair of vertices $u,v\in V$ such that $u<_\sigma v$ and $u \rightsquigarrow_G v$. 
By \Cref{lem:distancelabel}, there exists some path $P \in \cP[u] \cap \cP[v]$ and vertices $u' \in C(u,P)$ and $v' \in C(v, P)$ such that $d_G(u,u') + d_P(u',v') + d_G(v', v) \le (1+\eps)\cdot d_G(u,v)$.
Consider the hierarchical decomposition of $P$ to subpaths, and let $P'\subseteq P$ be the maximum subpath containing both $u',v'$. 
That is, $P'$ is a consecutive subpath of $P$, containing both $u,v$, such that once we delete the centroid $x(P')$, the (possibly) two remaining subpaths does not contain both $u',v'$.
It follows that the path $P'$ goes from $u'$ to $x(P')$ to $v'$. 
Note that $x(P')\in X_u\cap X_v$. 
See \Cref{fig:PathHirarchy} for an illustration.
Using \Cref{lem:vertexGadget} we conclude
\begin{align*}
d_{D_{\sigma}}(u,v)\le d_{D_{x(P')}}(u,v) & =d_{G}(u,x(P'))+d_{G}(x(P'),v)\\
 & \le d_{G}(u,u')+d_{G}(u',x(P'))+d_{G}(x(P'),v')+d_{G}(v',v)\\
 & \le d_{G}(u,u')+d_{P}(u',v')+d_{G}(v',v)\le(1+\eps)\cdot d_{G}(u,v)~.
\end{align*}

The DAG $D_{\sigma^{-1}}$ constructed in the exact same way, but with respect to  the order $\sigma^{-1}$. 
As for every $u,v$, either $u<_\sigma v$ or $u<_{\sigma^{-1}} v$ we conclude that $\min\{d_{D_\sigma}(u,v), d_{D_{\sigma^{-1}}}(u,v)\} \leq (1+\eps)d_G(u,v)$, as required.

\section{Discussion and Open Problems}\label{sec:conclusion}
In this paper we studied Steiner DAG covers, and demonstrated a stark separation between Steiner and non-Steiner DAG covers (see \Cref{tab:results}). However, many fascinating problems remain unresolved. We discuss several of them:

\smallskip
\begin{enumerate}
\item 
\textbf{Bounds for general graphs.}
Filtser \cite{Fil25dags} showed that any digraph admits a non-Steiner $(\tilde{O}(\log n),O(\log n),\tilde{O}(m))$-DAG cover. However, there is no lower bound showing that $(t,g,\tilde{O}(m))$-DAG cover is impossible even for constant $t$ and $g$. 
Because the approach from \cite{AHW25,Fil25dags} is based on directed low diameter decompositions (LDDs) \cite{BGW20,BNW25,HHWZ25,Li25}, and
because the Lipschitz parameter of directed LDDs for general digraphs has a provable lower bound of $\Omega(\log n)$ \cite{Bar96}, we would need a drastically different approach from \cite{AHW25,Fil25dags}.
Understanding the best possible DAG covers for general digraphs (in both the Steiner and non-Steiner settings) is a fascinating open question.

\item \textbf{Preserving the topological structure}. In both \Cref{thm:treewidth,thm:planar} our focus was on obtaining sparse DAGs in the cover. Since small treewidth or planarity are highly desirable graph properties, it is natural to ask that all the DAGs in the DAG cover of a bounded treewidth / planar graph to remain so.
For treewidth $\tw$ digraphs, our \Cref{thm:treewidth} constructs two DAGs with pathwidth $O(\tw\cdot\log n)$.
While pathwidth is a much more structured (and useful) property%
\footnote{E.g., (undirected) graphs with constant pathwidth embed into $\ell_2$ with constant stretch \cite{AFGN22}, while there is a graph with treewidth $2$ that requires stretch $\Omega(\sqrt{\log n})$ \cite{NR02}. Another example is universal Steiner tree, where bounded pathwidth graphs admit a solution with stretch $O(\log n)$ \cite{BCFHHR24,Fil24}, while nothing better than the stretch $O(\log^7n)$ of general graph is known for graphs of treewidth $2$.}, 
we would like a smaller blow-up in the treewidth.
The DAGs we construct for planar digraphs (\Cref{thm:planar}) are not planar at all.
Preserving planarity, or bounded treewidth, is an interesting question.

\item \textbf{Aspect ratio dependence.} Our planar result (\Cref{thm:planar}), has a logarithmic dependence on the \emph{aspect ratio}.
The aspect ratio can be unbounded, so it is desirable to remove this dependency.
For the much better studied $1+\eps$ approximate distance oracles for planar digraphs there is also a dependency on the aspect ratio (both in Thorup \cite{Tho04} and in the state-of-the-art oracle of Le and Wulff-Nilsen \cite{LW21}).
Removing this dependency is a nice open question.

\item \textbf{Exact planar DAG covers}. Our DAG cover for planar digraph has stretch $1+\eps$, while in the treewidth case distances are preserved exactly.
It is interesting to understand what is the minimum number of additional edges $\mu$ such that every planar digraph admits a Steiner $(1,2,\mu)$-DAG cover. 
This problem may look challenging: for undirected planar graphs, exact tree cover requires $\Omega(n^{\frac 1 3 - \eps})$ trees\footnote{If one demands a tree cover that doesn't use Steiner vertices or edges, then the lower bound can be improved to $\Omega(n^{1/2})$, which is tight \cite{GKR04}.}, due to an information-theoretic lower bound on distance labeling \cite{CCLMST24socg, GPPR04}.
However, the tree cover lower bound does not imply a DAG cover lower bound, as DAGs are more expressive than trees. Some hope comes from a recent result of Charalampopoulos\etal \cite{CGLMPWW23}, who constructed an exact distance oracle for planar digraphs with space $n^{1+o(1)}$ and $\polylog (n)$ query time.
Whether there is a Steiner $(1,2,n^{1+o(1)})$-DAG cover for planar digraphs is also an interesting question.
\end{enumerate}

\paragraph*{Acknowledgements.} This work was initiated at Dagstuhl Seminar 25212: Metric Sketching and Dynamic Algorithms for Geometric and Topological Graphs. We thank the organizers and other participants for a productive environment.

\newpage

{\small 
 \bibliographystyle{alpha}
\bibliography{LSObib}}

\newpage
\appendix

 \section{Pathwidth property in \Cref{thm:treewidth}}\label{sec:pathwidth}
This section is devoted to showing that there is a permutation $\sigma$ such that the DAGs $D_\sigma,D_{\sigma^{-1}}$ created during the proof of \Cref{thm:treewidth} have pathwidth $O(\tw\cdot \log n)$.
Our proof of \Cref{thm:treewidth} holds with respect to  an arbitrary permutation $\sigma$. To get DAGs with small pathwidth we will use a specific permutation. 
Let $B$ be the separator bag (as above), and let $C_1,\dots,C_g$ be the connected components of $G\setminus B$ (in the undirected sense). We will construct $\sigma$ so that the vertices of each connected component will appear consecutively along $\sigma$.
That is, the vertices of $B$ will be the $|B|$ first vertices in $\sigma$. Then the vertices of $C_1$, then $C_2$, and so on, where the vertices of $C_g$ will appear last.
The internal order $\sigma_i$ among the vertices of $C_i$ is determined recursively in the same manner.
That is, we find a separator bag $B_i$ in $G[C_i]$. The vertices of $B_i$ will appear first in $\sigma_i$, and afterwards, consecutively, the vertices of each connected component of $G[C_i]\setminus B_i$.
We will construct two DAGs $D_\sigma,D_{\sigma^{-1}}$ with respect to  the permutations $\sigma,\sigma^{-1}$.

Denote $\sigma=(v_1,\dots,v_n)$. 
Let $x\in B$ be some vertex in the separator bag and consider the DAG $D_x$ constructed in \Cref{lem:vertexGadget}.
Observe that $D_x$ has pathwidth $2$.
For $i\in[1,n-1]$, let $X_{x,i}=\{v_i,u_i,u_{i+1}\}$, and $X_{x,n}=\{v_n,u_n\}$.
Clearly, by the definition of $D_x$, $(X_{x,1},X_{x,2},\dots,X_{x,n})$ is a path decomposition of $D_x$. 
In particular, there is a one-to-one correspondence between vertices and bags, and the order of the bags along the path decomposition is $\sigma$, the order of the vertices in the permutation.
See \Cref{fig:pathwidth} of the path decomposition of the bidirected star from \Cref{fig:doubleStar}.

\begin{figure}
    \centering
    \includegraphics[width=0.98\textwidth]{ 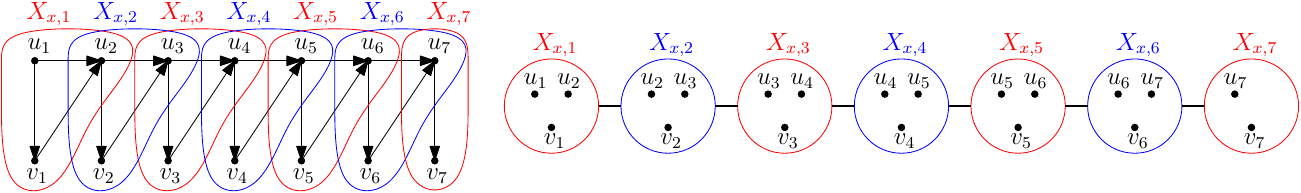}
    \caption{Bag decomposition of path construction for the biderected star.}
    \label{fig:pathwidth}
\end{figure}

The path decomposition for the entire DAG $D_\sigma$ is constructed recursively in the natural manner. Suppose by induction that for each connected component $C_j=\{v_\alpha,v_{\alpha+1},\dots,v_\beta\}$ of $G\setminus B$ there is a path decomposition $\cP_j=(X'_\alpha,\dots,X'_\beta)$ of $D_{\sigma_j}$, where each vertex $v_i\in C_j$ belongs only to the bag $X'_i$. Further, we assume that the width of $\cP_j$ is at most $2\cdot (\tw+1)\cdot \log |C_j|$. 
For each vertex $x\in B$ we create the DAG $D_x$ with the path decomposition $\left(X_{x,1},X_{x,2},\dots,X_{x,n}\right)$.
The path decomposition of $D_\sigma$ will be $\cP=(X_1,\dots,X_n)$, where the bag $X_i$ of $x_i\in C_j$ will consist of the union of $X'_i$ with $\{X_{x,j}\}_{x\in B}$ (the bag of each $x_j\in B$ is only the union of $X'_i$ with $\{X_{x,j}\}_{x\in B}$).
Note that each vertex of $D_\sigma$ is contained in a consecutive subset of bags in $\cP$. This is as each original vertex $v\in V$ is contained in a single bag, and each other vertex was contained in a consecutive subset, and is contained now in respective consecutive subset.    
Using the fact that each bag $X_{x,j}$ contains only two vertices beyond $v_j$, and the   inductive hypothesis, the width of the resulting path decomposition is 
\[  2|B|+2(\tw+1)\cdot\max_{j}\log|C_{j}|\le2(\tw+1)+2(\tw+1)\cdot\log\frac{n}{2}=2(\tw+1)\cdot\log n~.
\]
The bound on the pathwidth of $D_{\sigma^{-1}}$ follows by the exact same lines. The theorem now follows.

\end{document}